\documentclass[journal]{IEEEtran}
\usepackage{color}
\usepackage{mathrsfs}
\usepackage{amsfonts}
\usepackage{amsthm}
\theoremstyle{remark}

\newtheorem{corollary}{Corollary}
\newtheorem{theorem}{Theorem}
\newtheorem{lemma}{Lemma}
\newtheorem{remark}{Remark}

\usepackage{algorithm}
\usepackage{cite}
\usepackage{fancyhdr}
\usepackage{bbm}
\usepackage{datetime}

\ifCLASSINFOpdf
   \usepackage[pdftex]{graphicx}
\else
   \usepackage[dvips]{graphicx}
   \graphicspath{{../eps/}}
   \DeclareGraphicsExtensions{.eps}
\fi

\usepackage[cmex10]{amsmath}
\usepackage{algorithmic}
\begin{document}

\title{\vspace{0.4cm}Energy Efficient Distributed Coding for Data Collection in a Noisy Sparse Network}
\author{
\IEEEauthorblockN{Yaoqing Yang, Soummya Kar and Pulkit Grover\vspace{-0.4in}}
%
\thanks{This work was partially supported by the National Science Foundation under Grant CCF-1513936, NSF ECCS-1343324, NSF CCF-1350314, and by NSF grant ECCS-1306128.

Y. Yang, S. Kar and P. Grover are with the Department of Electrical and Computer Engineering, Carnegie Mellon University, Pittsburgh, PA, 15213, USA. Email: \{yyaoqing,soummyak,pgrover\}@andrew.cmu.edu}
}
\maketitle
\rfoot{}
\renewcommand{\headrulewidth}{0pt}


\begin{abstract}
We consider the problem of data collection in a two-layer network consisting of (1) links between $N$ distributed agents and a remote sink node; (2) a sparse network formed by these distributed agents. We study the effect of inter-agent communications on the overall energy consumption. Despite the sparse connections between agents, we provide an in-network coding scheme that reduces the overall energy consumption by a factor of $\Theta(\log N)$ compared to a naive scheme which neglects inter-agent communications. By providing lower bounds on both the energy consumption and the sparseness (number of links) of the network, we show that are energy-optimal except for a factor of $\Theta(\log\log N)$. The proposed scheme extends a previous work of Gallager~\cite{Gal_TIT_88} on noisy broadcasting from a complete graph to a sparse graph, while bringing in new techniques from error control coding and noisy circuits.
\end{abstract}

\vspace{-3mm}
\section{Introduction}
Consider a problem of collecting messages from $N$ distributed agents in a two-layer network. Each agent has one independent random bit $x_i\sim \text{Bernoulli}(\frac{1}{2})$, called the \emph{self-information bit}. The objective is to collect all self-information bits in a remote sink node with high accuracy. Apart from a noisy channel directly connected to the sink node, each agent can also construct a few noisy channels to other agents. We assume that, the inter-agent network has an advantage that an agent can transmit bits simultaneously to all its neighbors using a broadcast. However, constructing connections between distributed agents is difficult, meaning that the inter-agent network is required to be sparse.

Since agents are connected directly to the sink, there exists a simple scheme~\cite{Gal_TIT_88} which achieves polynomially decaying error probability with $N$: for all $n$ such that $1\le n\le N$, the $n$-th agent transmits $x_n$ to the sink for $\Theta(c\log N)$ times, where $c>1$, to ensure that $\Pr(\hat{x}_n\neq x_n)=\mathcal{O}\left(\frac{1}{N^c}\right)$. Then, using the union bound, we have that $\Pr(\hat{\mathbf{x}}\neq\mathbf{x})=\mathcal{O}\left(\frac{1}{N^{c-1}}\right)$. However, this naive scheme can only provide a solution in which the number of transmissions scales as $\Theta(N\log N)$. In this paper, we show that, by carrying out $\Theta(N\log\log N)$ inter-agent broadcasts, we can reduce the number of transmissions between distributed agents and the remote sensor from $\Theta(N\log N)$ to $\Theta(N)$, and hence dramatically reduce the energy consumption. Moreover, we show that, for the inter-agent broadcasting scheme to work, only $\Theta(N\log N)$ inter-agent connections are required.

A related problem is function computation in sensor networks~\cite{Gal_TIT_88,Gir_JSAC_05,Kara_TIT_11,Kow_TIT_12,Ying_TIT_07,Kamath_TIT_14}, especially the identity function computation problem \cite{Gir_JSAC_05,Kara_TIT_11,Gal_TIT_88}. In~\cite{Gal_TIT_88}, Gallager designed a coding scheme with $\mathcal{O}(N\log\log N)$ broadcasts for identify function computation in a complete graph. Here, we address the same problem in a much sparser graph and obtain the same scaling bound using a conceptually different distributed encoding scheme that we call \emph{graph code}. We also show that, the required inter-agent graph is the sparsest graph except for a $\Theta(\log\log N)$ factor, in that the number of links in the sparsest graph for achieving the $\mathcal{O}(N\log\log N)$ number of communications (energy consumption) has to be $\Omega\left(\frac{N\log N}{\log \log N}\right)$, if the error probability $\Pr(\hat{\mathbf{x}}\neq\mathbf{x})$ is required to be $o(1)$. In~\cite{Gir_JSAC_05}, Giridhar and Kumar studied the rate of computing \emph{type-sensitive} and \emph{type-threshold} functions in a random-planar network. In~\cite{Kara_TIT_11}, Karamchandani, Appuswamy and Franceschetti studied function computing in a grid network. Readers are referred to an extended version \cite{Yang_Arx_15_2} for a thorough literature review.

From the perspective of coding theory, the proposed graph code is closely related to erasure codes that have low-density generator matrices (LDGM). In fact, the graph code in this paper is equivalent to an LDGM erasure code with noisy encoding circuitry~\cite{Yang_All_14}, where the encoding noise is introduced by distributed encoding in the noisy inter-agent communication graph. Based on this observation, we show (in Corollary~\ref{coding_upb}) that our result directly leads to a known result in LDGM codes. Similar results have been reported by Luby~\cite{Lub_FOCS_02} for fountain codes, by Dimakis, Prabhakaran and Ramchandran \cite{Dim_TON_06} and by Mazumdar, Chandar and Wornell \cite{Maz_JSAC_14} for distributed storage, both with noise-free encoding. In the extended version~\cite{Yang_Arx_15_2}, we show that this LDGM code achieves sparseness (number of $1$'s in the generator matrix) that is within a $\Theta(\log\log N)$ multiple of an information-theoretic lower bound. Finally, We briefly summarize the main technical contributions of this paper:
\begin{itemize}
  \item we extend the classic distributed data collection problem (identity function computation) to sparse graphs, and obtain the same scaling bounds on energy consumption;
  \item we provide both upper and lower bounds on the sparseness (number of edges) of the communication graph for constrained energy consumption;
  \item we extend classic results on LDGM codes to in-network computing with encoding noise.
\end{itemize}

\section{System Model and Problem Formulations}\label{modeling}
Denote by $\mathcal{V}=\{v_1,\dots,v_N\}$ the set of distributed agents. Assume that in the first layer of the network, each agent has a link to the sink node $v_0$, and this link is a BEC (binary erasure channel) with erasure probability $\epsilon$. Each transmission from a distributed agent to the sink consumes energy $E_1$. We denote by $\mathcal{G}=(\mathcal{V},\mathcal{E})$ the second layer of the network, \emph{i.e.}, a directed inter-agent graph. We assume that each directed link in $\mathcal{G}$ is also a BEC with erasure probability $\epsilon$. We denote by $\mathcal{N}_v^-$ and $\mathcal{N}_v^+$ the one-hop in-neighborhood and out-neighborhood of $v$. Each broadcast from a node $v$ to all of its out-neighbors in $\mathcal{N}_v^+$ consumes energy\footnote{Due to possibly large distance to the sink node, it is likely that the energy consumption $E_2<E_1$. But we do not make any specific assumption on relationship between $E_2$ and $E_1$.} $E_2$. We allow $\mathcal{N}_v^-$ and $\mathcal{N}_v^+$ to contain $v$ itself (self-loops), because a node can broadcast information to itself. Denote by $d_n$ the out-degree of the $v_n$. Then, we have that $|\mathcal{E}|=\mathop\sum\limits_{n=1}^Nd_n$.
\vspace{-3mm}
\subsection{Data Gathering with Transmitting and Broadcasting}\label{Graph_Model}
A \emph{computation scheme} $\mathscr{S}=\{f_t\}_{t=1}^{T}$ is a sequence of Boolean functions, such that at each time slot $t$, a single node $v(t)$ computes the function $f_{t}$ (whose arguments are to be made precise below), and either broadcasts the computed output bit to $\mathcal{N}_v^+$, or transmits to $v_0$. We assume that the scheme terminates in finite time, \emph{i.e.}, $T<\infty$. The arguments of $f_t$ may consist of all the information that the broadcasting node $v(t)$ has up to time $t$, including its self-information bit $x_{v(t)}$, randomly generated bits and information obtained from its in-neighborhood. A scheme has to be feasible, meaning that all arguments of $f_t$ should be available at $v(t)$ before time $t$. We only consider oblivious transmission schemes, \emph{i.e.}, the three-tuple $(T,\{f_t\}_{t=1}^{T},\{v(t)\}_{t=1}^{T})$ and the decisions to broadcast or to transmit are predetermined. Denote by $\mathcal{F}$ the set of all feasible oblivious schemes. For a feasible scheme $\mathscr{S}\in\mathcal{F}$, denote by $t_{n,1}$ the number of transmissions from $v_n$ to the sink, and by $t_{n,2}$ the number of broadcasts from $v_n$ to $\mathcal{N}_v^+$. Then, the overall energy consumption is
\begin{equation}\label{Total_E}
  E=\mathop\sum_{n=1}^N E_1t_{n,1}+E_2t_{n,2}.
\end{equation}
Conditioned on the graph $\mathcal{G}$, The error probability is defined as $P_e^{\mathcal{G}}=\Pr(\hat{\mathbf{x}}\neq\mathbf{x})$, where $\hat{\mathbf{x}}$ denotes the final estimate of $\mathbf{x}$ at the sink $v_0$. It is required that $P_e^{\mathcal{G}}\le p_{\text{tar}}$ where $p_{\text{tar}}$ is the target error probability and might be zero. We also impose a sparse constraint on the problem, meaning the number of edges in the second layer of the network is smaller than $D$. The problem to be studied is therefore
\begin{equation}\label{op_problem}
\begin{split}
\text{\bf{Problem 1:   }}{{\min }_{\mathcal{G},\mathcal{S}\in \mathcal{F}}}\ \ E,\text{s}\text{.t}\text{.}\ \left\{ \begin{matrix}
   P_{e}^{\mathcal{G}}\le {{p}_{\text{tar}}},  \\
   |\mathcal{E}|<D.  \\
\end{matrix} \right.
\end{split}
\end{equation}
A related problem formulation is to minimize the number of edges (obtaining the sparsest graph) while making the energy consumption constrained:
\begin{equation}\label{op_problem_2}
\begin{split}
\text{\bf{Problem 2:   }}{{\min }_{\mathcal{G},\mathcal{S}\in \mathcal{F}}}\ \ |\mathcal{E}|,\text{s}\text{.t}\text{.}\ \left\{ \begin{matrix}
   P_{e}^{\mathcal{G}}\le {{p}_{\text{tar}}},  \\
   E<E_{M}.  \\
\end{matrix} \right.
\end{split}
\end{equation}
\subsection{Lower Bounds on Energy Consumption and Sparseness}
\begin{theorem}\label{best_degree}(Lower Bounds)
For Problem 1, suppose $\frac{{{N}^{2}}}{4\delta D}>{{e}^{1.5}}$, where $\delta=\ln \frac{1}{1-p_\text{tar}}=\Theta(p_\text{tar})$. Then, the solution of Problem 1 satisfies
\begin{equation}\label{p1lb}
\begin{split}
  &E\ge\max\left( NE_1,\frac{1}{\ln(1/\epsilon)}\min\left(\frac{N{{E}_{1}}}{2}\ln \frac{N}{2\delta },\frac{{{N}^{2}}{{E}_{2}}}{4D}\ln \frac{N}{2\delta }\right)\right)\\
  =&\Omega \left(\max\left(NE_1 ,\min \left( N{{E}_{1}}\ln \frac{N}{p_\text{tar}},\frac{{{N}^{2}}{{E}_{2}}}{D}\ln \frac{N}{p_\text{tar}} \right) \right)\right).
\end{split}
\end{equation}
For Problem 2, suppose $\frac{E_2N^2}{4\delta E_M}>e^{1.5}$ and $E_M< \frac{N{{E}_{1}}}{2\ln \left( 1/\epsilon  \right)}\ln \frac{N}{2\delta }$. Then, solution\footnote{Note that when the energy constraint $E_M\to 0$, the RHS of \eqref{p2lb} goes to infinity. This does not mean the lower bound is wrong, but means that Problem 2 does not have a feasible solution, and hence the minimized value of Problem 2 is infinity. See Remark~\ref{remark2} in Appendix~\ref{PofThm2} for details.} of Problem 2 satisfies
\begin{equation}\label{p2lb}
   |\mathcal{E}|\ge \frac{{{N}^{2}}{{E}_{2}}}{4\ln \left( 1/\varepsilon  \right){{E}_{M}}}\ln \frac{N}{2\delta }=\Omega\left({\frac{{{N}^{2}}{{E}_{2}}}{{{E}_{M}}}\ln \frac{N}{2p_\text{tar} }}\right).
\end{equation}
\end{theorem}
\begin{proof}
Due to limited space, we only include a brief introduction on the idea of the proof. See Appendix~\ref{PofThm2} for a complete proof. First, for the $n$-th node, the probability ${{p}_{\text{n}}}$ that all ${{t}_{n,1}}$ transmissions and ${{t}_{n,2}}$ broadcasts to its $d_n$ neighbors are erased is ${{p}_{n}}={{\epsilon }^{{{t}_{n,1}}+{{d}_{n}}{{t}_{n,2}}}}$. If this event happens for $v_n$, all information about $x_n$ is erased, and hence all self-information bits cannot be recovered. Thus,
\begin{equation}
  P_e^{\mathcal{G}}=\Pr(\hat{\mathbf{x}}\neq\mathbf{x})\ge1-\mathop\prod\limits_{n=1}^N (1-p_\text{n})=1-\mathop\prod\limits_{n=1}^N (1-{{\epsilon }^{{{t}_{n,1}}+{{d}_{n}}{{t}_{n,2}}}}).
\end{equation}
The above inequality can be relaxed by
\begin{equation}\label{cons_1_1}
  \sum\limits_{n=1}^{N}{{{\epsilon }^{{{t}_{n,1}}+{{d}_{n}}{{t}_{n,2}}}}}<\ln \frac{1}{1-P_{e}^{\mathcal{G}}}<\ln \frac{1}{1-p_\text{tar}},
\end{equation}
where $p_\text{tar}$ is the target error probability. The lower bounds of Problem 1 and Problem 2 are obtained by relaxing the constraint $P_{e}^{\mathcal{G}}<p_\text{tar}$ by \eqref{cons_1_1}. In what follows, we provide some intuition for Problem 1 as an example. For Problem 1, we notice that, in order to make the overall energy $E$ in \eqref{Total_E} smaller, we should either make $t_{n,1}$ smaller, or make $t_{n,2}$ smaller, while maintaining $t_{n,1}+d_nt_{n,2}$ large enough to make $\eqref{cons_1_1}$ hold. Actually, we can make the following observations:
\begin{itemize}
  \item if ${{d}_{n}}\le \frac{{{E}_{2}}}{{{E}_{1}}}$, we should set ${{t}_{n,2}}=0$, \emph{i.e.} we should forbid $v_n$ from broadcasting. Otherwise, we should set ${{t}_{n,1}}=0$;
  \item if ${{d}_{n}}\le \frac{{{E}_{2}}}{{{E}_{1}}}$, since ${{t}_{n,2}}=0$, we can always make the energy consumption $E$ smaller by setting ${{d}_{n}}=0$, \emph{i.e.}, we construct no out-edges from $v_n$ in the graph $\mathcal{G}$.
\end{itemize}
Using these observations, we can decompose the original optimization into two subproblems respectively regarding $d_n\ge E_2/E_1$ and $d_n< E_2/E_1$. We can complete the proof using standard optimization techniques and basic inequalities.
\end{proof}
\begin{remark}
Note that the lower bounds hold for individual graph instances with arbitrary graph topologies. Although the two lower bounds are not tight for all cases, we especially care about the case when the sparseness constraint $D$ satisfies $D=\mathcal{O}(N\log N)$ and the energy constraint $E_M$ satisfies $E_M=o(N\log N)$. In this case, we will provide an upper bound that differs from the lower bound by a multiple of $\Theta(\log\log N)$. In Section~\ref{Best_Degree}, we provide a detailed comparison between the upper and the lower bounds.
\end{remark}

\section{Main Technique: Graph Code}\label{main_technique}
In this section, we provide an distributed coding scheme in accordance with the goal of Problem 1 and Problem 2. The code considered in this paper, which we call $\mathcal{GC}$-3 graph code\footnote{We name this code $\mathcal{GC}$-3 because we also designed $\mathcal{GC}$-1 and $\mathcal{GC}$-2 graph codes. Readers are referred to an extended version of this paper \cite{Yang_Arx_15_2} for more details. Problem in \cite{Yang_Arx_15_2} are motivated from the perspective of communication complexity, which is fundamentally different from this paper.}, is a systematic binary code that has a generater matrix $\mathbf{G}=[\mathbf{I},\mathbf{A}]$ with $(\mathbf{A})_{N\times N}$ being the graph adjacency matrix of $\mathcal{G}$, \emph{i.e.}, $A_{i,j}=1$ if there is a directed edge from $v_i$ to $v_j$. The encoding of the $\mathcal{GC}$-3 graph code can be written as
\begin{equation}\label{3g_code}
  \mathbf{r}^\top=\mathbf{x}^\top \cdot\left[\mathbf{I},\mathbf{A}\right],
\end{equation}
\textcolor{black}{where $\mathbf{x}^\top=[x_1,x_2,\dots,x_N]$ denotes the self-information bits and $\mathbf{r}^\top$ denotes the encoding output with length $2N$.} \textcolor{black}{This means that the code bit calculated by a node $v$ is either its self-information bit $x_v$ or the parity of the self-information bits in its in-neighborhood $\mathcal{N}_v^-$. Therefore, $\mathcal{GC}$-3 codes are easy to encode using inter-agent broadcasts and admit distributed implementations.} In what follows, we define the in-network computing scheme associated with the $\mathcal{GC}$-3 code.
\subsection{In-network Computing Scheme}\label{Algorithm}
The in-network computing scheme has two steps. During the first step, each node take turns to broadcast its self-information bit to $\mathcal{N}^+(v)$ for $t$ times, where
\begin{equation}\label{Broadcasting_Times}
  t=\frac{1}{\log(1/\epsilon)}{\log\left(\frac{c\log N}{p_{\text{ch}}}\right)},
\end{equation}
where $c\in(0,\infty)$ and $p_{\text{ch}}\in(0,1/2)$ are two predetermined constants. Then, each node estimates all self-information bits from all its in-neighbors in $\mathcal{N}_v^-$. The probability that a certain bit is erased for $t$ times when transmitted from a node $v$ to one of its out-neighbors is
\begin{figure}
  \centering
 \includegraphics[scale=0.3]{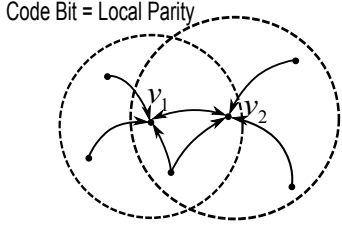}\\
  \caption{\emph{Each code bit is the parity of all one-hop in-neighbors of a specific node. Some edges in the directed graph might be bi-directional.\vspace{-5mm}}}\label{L_parity}
\end{figure}
\begin{equation}\label{one_bit_error_equation}
  P_e=\epsilon^t = \frac{p_{\text{ch}}}{c\log N}.
\end{equation}

If all information bits from its in-neighborhood $\mathcal{N}^-(v_n)$ are sent successfully, $v_n$ computes the local parity
\begin{equation}\label{Local_Parity}
  y_n=\mathop\sum \limits_{v_m\in \mathcal{N}^-(v_n)}x_m=\mathbf{x}^\top\mathbf{a}_n,
\end{equation}
where $\mathbf{a}_n$ is the $n$-th column of the adjacency matrix $\mathbf{A}$, and the summation is in the sense of modulo-2. If any bit $x_m$ is not sent to $v_n$ successfully, \emph{i.e.}, erased for $t$ times, the local parity cannot be computed. In this case, $y_n$ is assumed to take the value `$e$'. We denote the vector of all local parity bits by $\mathbf{y}=[y_1,y_2,...,y_N]^\top$. If all nodes could successfully receive all information from their in-neighborhood, we would have
\begin{equation}\label{Local_Parity_Encoding}
  \mathbf{y}^\top=\mathbf{x}^\top \mathbf{A},
\end{equation}
where $A$ is the adjacency matrix of the graph $\mathcal{G}$.

During the second step, each node $v_n$ transmits $x_n$ and the local parity $y_n$ to the sink exactly once. If a local parity $y_n$ has value `$e$', $v_n$ sends the value `$e$'. Denote the received (possibly erased) version of the self-information bits at the sink by $\tilde{\mathbf{x}}=[\tilde{x}_1,\tilde{x}_2,...,\tilde{x}_N]^\top$, and the received (possibly erased) version of local parities by $\tilde{\mathbf{y}}=[\tilde{y}_1,...,\tilde{y}_N]$. Notice that, there might be some bits in $\mathbf{y}$ changed into value `$e$' during the second step. We denote all information gathered at the sink by $\mathbf{r}=[\tilde{\mathbf{x}}^\top,\tilde{\mathbf{y}}^\top]$. If all the connections between the distributed agents and from the distributed agents to the sink were perfect, the received information $\mathbf{r}$ at the sink could be written as~\eqref{3g_code}. However, the received version is possibly with erasures, so the sink carries out the Gaussian elimination algorithm to recover all information bits, using all non-erased information. If there are too many erased bits, leading to more than one possible decoded values $\hat{\mathbf{x}}^\top$, the sink claims an error.

In all, the energy consumption is
\begin{equation}\label{Low_Diam_Up_Bd}
\begin{split}
  E=&2N\cdot E_1+N\cdot t\cdot E_2
  =2NE_1+\frac{N\log(\frac{c\log N}{p_{\text{ch}}})}{\log(1/\epsilon)}E_2\\
  =&\Theta\left(\max\left(NE_1,NE_2\log\log N\right)\right),
\end{split}
\end{equation}
where $t$ is defined in~\eqref{Broadcasting_Times}, and the constant $2$ in $2N\cdot E_1$ is introduced in the second step, when both the self-information bit and the local parity are transmitted to the sink.

\section{Analysis of the Error Probability}\label{Upper_Bound}
First, we define a random graph ensemble based on the Erd$\ddot{o}$s-R$\acute{e}$nyi graphs~\cite{Bol_Spr_98}. In this graph ensemble, each node has a directed link to another node with probability $p=\frac{c\log N}{N}$, where $c$ is the same constant in~\eqref{Broadcasting_Times}. All connections are independent of each other. We sample a random graph from this graph ensemble and carry out the in-network broadcasting scheme provided in Section~\ref{Algorithm}. Then, the error probability $P_e^{\mathcal{G}}(\mathbf{x})$ is itself a random variable, because of the randomness in the graph sampling stage and the randomness of the input. We define $P_e^{(N)}(\mathbf{x})$ as the expected error probability $P_e^{(N)}(\mathbf{x})=\mathbb{E}_\mathcal{G} [P_e^{\mathcal{G}}(\mathbf{x})]$ over the random graph ensemble.

\begin{theorem}\label{Theorem1}\label{ER_main_thm} (Upper Bound on the Ensemble Error Probability)
Suppose $\eta>0$ is a constant, $p_{\text{ch}}\in (0,\frac{1}{2})$ is a constant, $\epsilon$ is the channel erasure probability and $\varepsilon_0=(\frac{2}{1-1/e}+1)p_{\text{ch}}+\epsilon$. Assume $c\log N>1$. Define
\begin{equation}\label{b}
  b_{\eta} = \frac{1}{2}(1 - {\varepsilon _0})(1 - \frac{{1 - {e^{ - 2c\eta }}}}{2}),
\end{equation}
and assume
\begin{equation}\label{N_big_enough}
  \epsilon<b_\eta.
\end{equation}
Then, for the transmission scheme in Section~\ref{Algorithm}, we have
\begin{equation}\label{error_exponent}
P_e^{(N)}(\mathbf{x}) \le {(1 - b_{\eta})^N}{\rm{ + }}\eta e \epsilon \frac{{{N^{2 - c(1 - {\varepsilon _0})(1 - c\eta )}}}}{{\log N}},\forall \mathbf{x}.
\end{equation}
That is to say, if $2 < c(1 - {\varepsilon _0})(1 - c\eta )$, the error probability eventually decreases polynomially with $N$. The rate of decrease can be maximized over all $\eta$ that satisfies~\eqref{N_big_enough}.
\end{theorem}
\begin{proof}
See Section~\ref{Analysis}.
\end{proof}
Thus, we have proved that the expected error probability averaged over the graph code ensemble decays polynomially with $N$. Denote by $A_e$ the event that an estimate error occurs at the sink, \emph{i.e.}, $\hat{\mathbf{x}}\neq \mathbf{x}$, then
\begin{equation}\label{rg_argu}
\begin{split}
  P_e^{(N)} >\Pr(2cN\log N>|\mathcal{E}|)\Pr\left(A_e\mid 2cN\log N> |\mathcal{E}|\right).
\end{split}
\end{equation}
Since the number of edges $|\mathcal{E}|$ in the directed graph is a Binomial random variable$\sim\text{Binomial}(p=\frac{c\log N}{N},N^2)$, using the Chernoff bound \cite{Che_AMS_52}, we can get
\begin{equation}
\begin{split}
  \Pr\left(2c N \log N>|\mathcal{E}|\right)\ge 1-\left(\frac{1}{N}\right)^{\frac{c^2}{2}\log N}.
\end{split}
\end{equation}
Combining with~\eqref{rg_argu} and \eqref{error_exponent},
\begin{equation}
  \Pr\left(A_e\mid 2cN\log N> |\mathcal{E}|\right)<\left(1-\left(\frac{1}{N}\right)^{\frac{c^2}{2}\log N}\right)^{-1}P_e^{(N)},
\end{equation}
which decays polynomially with $N$. This means that there exists a graph code (graph topology) with $\mathcal{O}(N\log N)$ links, and at the same time, achieves any required non-zero error probability $p_\text{tar}$ when $N$ is large enough.
Interestingly, the derivation above implies a more fundamental corollary for erasure coding in point-to-point channels. The following corollary states the result for communication with noise-free circuitry, while the conclusions in this paper (see Theorem~\ref{ER_main_thm}) shows the existence of an LDGM code that is tolerant of noisy encoding and distributed encoding.
\begin{corollary}\label{coding_upb}
For a discrete memoryless \textcolor{black}{point-to-point} BEC with erasure probability $\epsilon$, there exists a systematic linear code with rate\footnote{Generalizing the analysis technique in this paper to $R>\frac{1}{2}$ is trivial, but designing a distributed encoding scheme for the inter-agent graph with $R>\frac{1}{2}$ is not intuitive. For $R=\frac{1}{2}$, each node sends its self-information bit and the local parity, which is practically convenient.} $R=1/2$ and an $N\times 2N$ generator matrix $\mathbf{G}=[\mathbf{I},\mathbf{A}]$ such that the block error probability decreases polynomially with $N$. Moreover, the generator matrix is sparse: the number of ones in $\mathbf{A}$ is $\mathcal{O}(N\log N)$.
\end{corollary}
\begin{proof}
See Appendix~\ref{pf_coding_upb}.
\end{proof}
\begin{remark}
In an extended version \cite[Section VI]{Yang_Arx_15_2}, we discuss a distributed coding scheme, called $\mathcal{GC}$-2, for a geometric graph. The $\mathcal{GC}$-2 code divides the geometric graph into clusters and conquer each cluster using a dense code with length $\mathcal{O}(\log N)$. Notice that the $\mathcal{GC}$-2 code requires the same sparsity $\Theta(N\log N)$ and the same number of broadcasts (and hence the same scale in energy consumption) as $\mathcal{GC}$-3. However, the scheduling cost of $\mathcal{GC}$-2 is high. Further, it requires a powerful code with length $\mathcal{O}(\log N)$, which is not practical for moderate $N$ (this is also the problem of the coding scheme in \cite{Gal_TIT_88}). Nonetheless, the graph topology for the $\mathcal{GC}$-2 code is deterministic, which does not require ensemble-type arguments.
\end{remark}
\subsection{Gap Between the Upper and the Lower Bounds}\label{Best_Degree}
In this part, we compare the energy consumption and the graph sparseness of the $\mathcal{GC}$-3 graph code with the two lower bounds in Theorem~\ref{best_degree}. First, we examine Problem 1 when $D=\Theta(N\log N)$ and $p_\text{tar}=\Theta\left(\frac{1}{N^\gamma}\right),\gamma\in(0,1)$, which is the same case as the $\mathcal{GC}$-3 Graph Code. In this case, the lower bound \eqref{p1lb} has the following form:
\begin{equation}
\begin{split}
  E=\Omega \left(\max\left(NE_1 ,\min \left( N{{E}_{1}}\log N,NE_2\right) \right)\right).
\end{split}
\end{equation}
Under the mild condition $\frac{E_2}{E_1}>\frac{1}{\log N}$, the lower bound can be simplified as
\begin{equation}
  E^\text{lower}=\Omega \left(\max\left(NE_1,NE_2\right)\right).
\end{equation}
The energy consumption of the $\mathcal{GC}$-3 graph code has the form $E^\text{upper}=\Theta \left(\max\left(NE_1,NE_2\log\log N\right)\right)$ (see~\eqref{Low_Diam_Up_Bd}), which has a $\Theta(\log\log N)$ multiplicative gap with the lower bound. Notice that if we make the assumption $E_1>>E_2$, \emph{i.e.}, the inter-agent communications are cheaper, the two bounds have the same scaling $\Theta(NE_1)$.

Then, we examine Problem 2 when $E_M=\Theta \left(\max\left(NE_1,NE_2\log\log N\right)\right)$ and $p_\text{tar}=\Theta\left(\frac{1}{N^\gamma}\right),\gamma\in(0,1)$, which is also the same case as the $\mathcal{GC}$-3 Graph Code. Notice that under mild assumptions, $E_M=\Theta \left(\max\left(NE_1,NE_2\log\log N\right)\right)=o(E_1N\log N)$, which means that the condition $E_M< \frac{N{{E}_{1}}}{2\ln \left( 1/\epsilon  \right)}\ln \frac{N}{2\delta }$ in Theorem~\ref{best_degree} holds when $N$ is large enough. In this case, the lower bound \eqref{p2lb} takes the form
\begin{equation}
\begin{split}
  |\mathcal{E}|=\Omega\left(\min\left({\frac{{{N}}{{E}_{2}}}{{{E}_{1}}}\log N},{\frac{{{N}}}{\log\log N}\log N}\right)\right).
\end{split}
\end{equation}
The number of edges of the $\mathcal{GC}$-3 graph code has the scale $|\mathcal{E}|=\Theta(N\log N)$. Therefore, the ratio between the upper and the lower bound satisfies that
\begin{equation}
  \frac{|\mathcal{E}^\text{upper}|}{|\mathcal{E}^\text{lower}|}=\mathcal{O}\left(\max(\log\log N,E_1/E_2)\right).
\end{equation}

\subsection{An Upper Bound on the Error Probability}\label{Analysis}
The Lemma~\ref{identical_error} in the following states that $P_e^{\mathcal{G}}(\mathbf{x})$ is upper bounded by an expression which is independent of the input $\mathbf{x}$ (self-information bits). In Lemma~\ref{identical_error}, each term on the RHS of~\eqref{identical_error_equation} can be interpreted as the probability of the existence of a non-zero vector input $\mathbf{x}_0$ that is confused with the all-zero vector $\mathbf{0}_N$ after all the non-zero entries of $\mathbf{x}_0^\top\cdot[\mathbf{I},\mathbf{A}]$ are erased, in which case $\mathbf{x}_{0}$ is indistinguishable from the all zero channel input. For example, suppose the code length is $2N=6$. The sent codeword $\mathbf{x}_0^\top\cdot[\mathbf{I},\mathbf{A}]=[x_1,0,0,x_4,0,x_6]$ and the output at the sink happens to be $\mathbf{r}^\top=[e,0,0,e,0,e]$. In this case, we cannot distinguish between the input vector $\mathbf{x}_0$ and $\mathbf{0}^N$ based on the output at the sink.

\begin{lemma}\label{identical_error}
The error probability $P_e^{\mathcal{G}}$ can be upper-bounded by
\begin{equation}\label{identical_error_equation}
  P_e^{\mathcal{G}}(\mathbf{x})\le \mathop \sum \limits_{\mathbf{x}_0 \in \{0,1\}^N\setminus \{\mathbf{0}_N\}} P_e^{\mathcal{G}}(\mathbf{x}_0\rightarrow \mathbf{0}^N),\forall{\mathbf{x}\in \{0,1\}^N},
\end{equation}
where $\mathbf{0}_N$ is the $N$-dimensional zero vector.
\end{lemma}
\begin{proof}
See Appendix~\ref{pf_Lemma2}.
\end{proof}

Therefore, to upper-bound $P_e^\mathcal{G}(\mathbf{x})$, we only need to consider the event mentioned above, \emph{i.e.}, a non-zero input $\mathbf{x}_0$ of self-information bits is confused with the all-zero vector $\mathbf{0}^N$. This happens if and only if each entry of the received vector $\mathbf{r}^\top$ at the sink is either zero or `$e$'. When $\mathbf{x}_0$ and the graph $\mathcal{G}$ are both fixed, different entries in $\mathbf{r}^\top$ are independent of each other. Thus, the ambiguity probability $P_e^{\mathcal{G}}(\mathbf{x}_0\rightarrow \mathbf{0}^N)$ for a fixed non-zero input $\mathbf{x}_0$ and a fixed graph instance $\mathcal{G}$ is the product of the corresponding ambiguity probability of each entry in $\mathbf{r}^\top$ (being a zero or a `$e$').

The ambiguity event of each entry may occur due to structural deficiencies in the graph topology as well as due to erasures. In particular, three events contribute to the error at the $i$-th entry of $\mathbf{r}^\top$: the product of $\mathbf{x}_0^\top$ and the $i$-th column of $[\mathbf{I},\mathbf{A}]$ is zero (topology deficiency); the $i$-th entry of $\mathbf{r}^\top$ is `$e$' due to erasures in the first step; the $i$-th entry is `$e$' due to an erasure in the second step. We denote these three events respectively by $A_{1}^{(i)}(\mathbf{x}_0)$, $A_{2}^{(i)}(\mathbf{x}_0)$ and $A_{3}^{(i)}(\mathbf{x}_0)$, where the superscript $i$ and the argument $\mathbf{x}_0^\top$ mean that the events are for the $i$-th entry and conditioned on a fixed message vector $\mathbf{x}_0^\top$. The ambiguity event on the $i$-th entry is the union of the above three events. Denote by the union event as $A^{(i)}(\mathbf{x}_0)=A_{1}^{(i)}(\mathbf{x}_0)\cup A_{2}^{(i)}(\mathbf{x}_0)\cup A_{3}^{(i)}(\mathbf{x}_0)$. By applying the union bound over all possible inputs, the error probability $P_e^{\mathcal{G}}(\mathbf{x})$ (for an arbitrary input $\mathbf{x}$) can be upper bounded by
\begin{equation}\label{Error_pre}
\begin{split}
  P_e^{\mathcal{G}}(\mathbf{x})\le \mathop \sum \limits_{\mathbf{x}_0 \in \{0,1\}^N\setminus\{\mathbf{0}^N\}} \mathop\prod \limits_{i=1}^{2N} \Pr[A^{(i)}(\mathbf{x}_0)|\mathcal{G}],
\end{split}
\end{equation}
In this expression, the randomness of $\mathcal{G}$ lies in the random edge connections. We use the binary indicator $E_{mn}$ to denote if there is a directed edge from $v_m$ to $v_n$. Note that we allow self-loops. By assumption, all random variables in $\{E_{mn}\}_{m,n=1}^N$ are mutually independent\footnote{Note a bidirectional edge in the current setting corresponds to two independently generated directional edges.}.
Therefore
\begin{equation}\label{6142}
\begin{split}
  &P_e^{(N)}(\mathbf{x})=\mathbb{E}_\mathcal{G} [P_e^{\mathcal{G}}(\mathbf{x})]\\
  \overset{(a)}{\le}&\mathop \sum \limits_{\mathbf{x}_0^\top \in \{0,1\}^N\setminus\{\mathbf{0}^N\}}
  \mathop\prod \limits_{i=1}^{2N} \mathbb{E}_\mathcal{G}\left[\Pr\left[A^{(i)}(\mathbf{x}_0)\left|E_{ni},1\le n\le N\right.\right]\right]\\
  \overset{(b)}{=}&\mathop \sum \limits_{\mathbf{x}_0^\top \in \{0,1\}^N\setminus\{\mathbf{0}^N\}}
  \mathop\prod \limits_{i=1}^{2N} \Pr[A^{(i)}(\mathbf{x}_0)],
\end{split}
\end{equation}
where the equality (a) holds because in the in-network computing scheme, the self-information bit $x_i$ and the local parity bit $y_i$ only depend on the in-edges of $v_i$, \emph{i.e.}, the edge set $\mathcal{E}_i^{\text{in}}=\{E_{ni}|1\le n\le N\}$, and the fact that different in-edge sets $\{E_{ni}\}_{1\leq n\leq N}$ and $\{E_{nj}\}_{1\leq n\leq N}$ are independent (by the independence of link generation) for any pair $(i,j)$ with $i\neq j$, and the equality (b) follows from the iterative expectation.
\begin{lemma}\label{Lemma1}
Define $k$ as the number of ones in $\mathbf{x}_0^\top$ and $\varepsilon_0=(\frac{2}{1-1/e}+1)p_{\text{ch}}+\epsilon$, where $\epsilon$ is the erasure probability of the BECs and $p_{\text{ch}}$ is a constant defined in~\eqref{Broadcasting_Times}. Further suppose $c\log N>1$. Then, for $1\le i\le N$, it holds that
\begin{equation}\label{Each_error_decompose_2}
  \mathop\prod \limits_{i=1}^{N}\Pr[A^{(i)}(\mathbf{x}_0)]= \epsilon^k.
\end{equation}
For $N+1\le i\le 2N$, it holds that
\begin{equation}\label{Each_error_decompose}
\begin{split}
  \Pr[A^{(i)}(\mathbf{x}_0)]\le \varepsilon_0+(1-\varepsilon_0)\cdot\frac{1+(1-2p)^k}{2},
\end{split}
\end{equation}
where $p=\frac{c\log N}{N}$ is the connection probability.
\end{lemma}
\begin{proof}
See Appendix~\ref{PofL1} for a complete proof. The main idea is to directly compute the probabilities of three error events $A_1^{(i)}$, $A_2^{(i)}$ and $A_3^{(i)}$ for each bit $x_i$.
\end{proof}

Based on Lemma~\ref{Lemma1} and simple counting arguments, note that~\eqref{6142} may be bounded as
\begin{equation}\label{Error_Middle}
  P_e^{(N)}(\mathbf{x})\le \mathop\sum\limits_{k=1}^N\binom{N}{k}\epsilon^k \left[\varepsilon_0+(1-\varepsilon_0)\cdot\frac{1+(1-2p)^k}{2}\right]^N.
\end{equation}
By upper-bounding the RHS of \eqref{Error_Middle} respectively for $k=\mathcal{O}(N/\log N)$ and $k=\Omega(N/\log N)$, we obtain Theorem~\ref{Theorem1}. The remaining part of the proof can be found in Appendix~\ref{PofT1}.

\section{Conclusions}
In this paper, we obtain both upper and lower scaling bounds on the energy consumption and the number of edges in the inter-agent broadcast graph for the problem of data collection in a two-layer network. In the directed Erd$\ddot{o}$s-R$\acute{e}$nyi graph ensemble, the average error probability of the proposed distributed coding scheme decays polynomially with the size of the graph. We show that the obtained code is almost optimal in terms of sparseness (with minimum number of ones in the generator matrix) except for a $\Theta(\log\log N)$ multiple gap. Finally, we show a connection of our result to LDGM codes with noisy and distributed encoding.

\bibliographystyle{ieeetr}
\bibliography{rough}
\appendices

\section{Proof of Theorem~\ref{best_degree}}\label{PofThm2}
First, we state a lemma that we will use in the proof.
\begin{lemma}\label{Theo2_lemma}
Suppose the constants $\delta ,\varepsilon \in \left( 0,\frac{1}{2} \right),A,N>0$. Suppose $\frac{{{N}^{2}}}{\delta A}>{{e}^{1.5}}$, and suppose the minimization problem
\begin{align}
\underset{\mathbf{x}\in \mathbb{R}^N,\mathbf{a}\in \mathbb{R}^N}{\mathop{\min }}\,\sum\limits_{n=1}^{N}{{{x}_{i}}},
\text{s}\text{.t}\text{. }\left\{ \begin{matrix}
   \sum\limits_{n=1}^{N}{{{a}_{n}}}\le A, {{a}_{n}}\ge 1,\forall n. \\
   \sum\limits_{n=1}^{N}{{{\epsilon }^{{{a}_{n}}{{x}_{n}}}}}<\delta ,
\end{matrix} \right.\end{align}
has a solution, \emph{i.e.}, the feasible region is not empty. Then, the solution of the above minimization problem satisfies that
\begin{equation}
  \sum\limits_{n=1}^{N}{x_{i}^{*}}\ge \frac{{{N}^{2}}}{A\ln \left( 1/\epsilon  \right)}\ln \frac{N}{\delta }.
\end{equation}
\end{lemma}
\begin{proof}
First, consider the case when ${{a}_{1}},\ldots ,{{a}_{N}}\ge 1$ are fixed. In this case, it can be easily shown in the KKT conditions that the minimization is obtained when
\[{{\epsilon }^{{{a}_{n}}{{x}_{n}}}}=\frac{\delta }{{{a}_{n}}\sum\limits_{l=1}^{N}{\frac{1}{{{a}_{l}}}}},\]
which is equivalent to
\begin{equation}\label{xge}
\begin{split}
  {{x}_{n}}=\frac{1}{\ln \left( 1/\epsilon  \right)}\frac{1}{{{a}_{n}}}\ln {{a}_{n}}+\frac{1}{\ln \left( 1/\epsilon  \right)}\frac{1}{{{a}_{n}}}\ln \left( \frac{1}{\delta }\sum\limits_{l=1}^{N}{\frac{1}{{{a}_{l}}}} \right).
\end{split}
\end{equation}
Since $\sum\limits_{n=1}^{N}{{{a}_{n}}}\le A,$ we have that $\sum\limits_{n=1}^{N}{\frac{1}{{{a}_{n}}}}\ge \frac{{{N}^{2}}}{A}.$ Therefore, for fixed ${{a}_{1}},\ldots ,{{a}_{N}}\ge 1$, summing up~\eqref{xge} for all $n$ and plug in $\sum\limits_{n=1}^{N}{\frac{1}{{{a}_{n}}}}\ge \frac{{{N}^{2}}}{A}$, we get
\begin{equation}
\begin{split}
\sum\limits_{n=1}^{N}{{{x}_{n}}}\ge &\frac{1}{\ln \left( 1/\epsilon  \right)}\sum\limits_{n=1}^{N}{\frac{1}{{{a}_{n}}}\ln {{a}_{n}} }+\frac{1}{\ln \left( 1/\epsilon  \right)}\ln \left( \frac{{{N}^{2}}}{\delta A} \right)\sum\limits_{n=1}^{N}{\frac{1}{{{a}_{n}}}}\\
=& \frac{1}{\ln \left( 1/\epsilon  \right)}\sum\limits_{n=1}^{N}{\frac{1}{{{a}_{n}}}\left(\ln {{a}_{n}}+ \ln \frac{{{N}^{2}}}{\delta A} \right)}.
\end{split}
\end{equation}
When $B:=\frac{{{N}^{2}}}{\delta A}>{{e}^{1.5}}$, we can prove that the function $f(x)=\frac{1}{x}\left( B+\ln x \right)$ is convex in $[1,\infty]$. Therefore, the function $\frac{1}{{{a}_{n}}}\left( B+\ln {{a}_{n}} \right)$ is convex in ${{a}_{n}}$. Using the Jensen's inequality, we have that
\begin{equation}
  \sum\limits_{n=1}^{N}{{{x}_{n}}}\ge \frac{1}{\ln \left( 1/\epsilon  \right)}N\cdot \frac{N}{A}\left( B+\ln \frac{A}{N} \right)=\frac{{{N}^{2}}}{A\ln \left( 1/\epsilon  \right)}\ln \frac{N}{\delta }.
\end{equation}\end{proof}
For the $n$-th node, the probability ${{p}_{\text{n}}}$ that all ${{t}_{n,1}}$ transmissions and ${{t}_{n,2}}$ broadcasts are erased is lower bounded by
\begin{equation}\label{pn_l}
  {{p}_{n}}>{{\epsilon }^{{{t}_{n,1}}+{{d}_{n}}{{t}_{n,2}}}}.
\end{equation}
If this event happens for any node, all instant messages cannot be computed reliably, because at least all information about $x_n$ is erased. Thus, we have
\begin{equation}
  P_e^{\mathcal{G}}>1-\mathop\prod\limits_{n=1}^N (1-p_\text{n}),
\end{equation}
which is equivalent to $P_e^{\mathcal{G}}>1-\mathop\prod\limits_{n=1}^N (1-p_\text{n})$. Using the AM-GM inequality, we have that
\begin{equation}
\begin{split}
  1-P_{e}^{\mathcal{G}}<{{\left[ \frac{1}{N}\sum\limits_{n=1}^{N}{(1-{{p}_{n}})} \right]}^{N}}={{\left( 1-\frac{1}{N}\sum\limits_{n=1}^{N}{{{p}_{n}}} \right)}^{N}},
\end{split}
\end{equation}
Using the fact that $1-x\le\exp(-x)$, we have that
\begin{equation}
  \frac{1}{N}\sum\limits_{n=1}^{N}{{{p}_{n}}}<1-{{e}^{-\frac{1}{N}\ln \frac{1}{1-P_{e}^{\mathcal{G}}}}}<\frac{1}{N}\ln \frac{1}{1-P_{e}^{\mathcal{G}}}.
\end{equation}
Plugging in~\eqref{pn_l}, we get
\begin{equation}\label{cons_1}
  \sum\limits_{n=1}^{N}{{{\epsilon }^{{{t}_{n,1}}+{{d}_{n}}{{t}_{n,2}}}}}<\ln \frac{1}{1-P_{e}^{\mathcal{G}}}<\ln \frac{1}{1-p_\text{tar}},
\end{equation}
where $p_\text{tar}$ is the target error probability in Problem 1 and Problem 2. Note that to provide a lower bound for solutions of Problem 1 and Problem 2, we can always replace a constraint with a relaxed version. In the following proof, we always relax the constraint $P_{e}^{\mathcal{G}}\le {{p}_{\text{tar}}}$ by~\eqref{cons_1}, which only makes our lower bound loose, but still legitimate.

Consider Problem 1, in which we have a constraint on the sparseness $\sum\limits_{n=1}^{N}{{{d}_{n}}}\le D$, and a constraint on the error probability $p$. Our goal is to minimize $E=\sum\limits_{n=1}^{N}{{{E}_{1}}{{t}_{n,1}}+{{E}_{2}}{{t}_{n,2}}}$. Note that in this problem, we have the constraint that ${{t}_{n,1}},{{t}_{n,2}},{{d}_{n}}\in {{\mathbb{Z}}^{+}}\cup \{0\},\forall n$. We relax this constraint to ${{d}_{n}}\in [1,\infty ]\cup \{0\},{{t}_{n,1}},{{t}_{n,2}}\in [0,\infty ],\forall n$, which still yields a legitimate lower bound.

First, we notice the following facts:
\begin{itemize}
  \item If ${{d}_{n}}\le \frac{{{E}_{2}}}{{{E}_{1}}}$, we should set ${{t}_{n,2}}=0$. Otherwise, we should set ${{t}_{n,1}}=0$.
  \item If ${{d}_{n}}\le \frac{{{E}_{2}}}{{{E}_{1}}}$, we can always make the energy consumption $E$ smaller by setting ${{d}_{n}}=0$.
\end{itemize}
\begin{proof}
For the $n$-th node, if we keep ${{t}_{n}}:={{t}_{n,1}}+{{d}_{n}}{{t}_{n,2}}$ fixed, the LHS of the constraint~\eqref{cons_1} does not change. Noticing that the energy spent at the $n$-th node can be written as ${{E}_{1}}{{t}_{n,1}}+{{E}_{2}}{{t}_{n,2}}={{E}_{1}}{{t}_{n}}+\left( {{E}_{2}}-{{E}_{1}}{{d}_{n}} \right){{t}_{n,2}}$, we arrive at the conclusion that we should set ${{t}_{n,2}}=0$ when ${{d}_{n}}\le \frac{{{E}_{2}}}{{{E}_{1}}}$. Otherwise, we should maximize ${{t}_{n,2}}$, which means setting ${{t}_{n,1}}=0$. This concludes the first statement.

Based on the first statement, we have that, when ${{d}_{n}}\le \frac{{{E}_{2}}}{{{E}_{1}}}$, we set ${{t}_{n,2}}=0$. Therefore, the constraint~\eqref{cons_1} does not contain ${{d}_{n}}$ for ${{d}_{n}}\le \frac{{{E}_{2}}}{{{E}_{1}}}$ anymore, which means that further reducing ${{d}_{n}}$ does not affect the constraints. Thus, we should set ${{d}_{n}}=0$, which can help relax the constraints for other ${{d}_{n}}$.
\end{proof}

We assume, W.L.O.G., ${{d}_{1}}\ge {{d}_{2}}\ge \cdots {{d}_{N}}\ge 0$. Using the two arguments above, we can arrive at the following statement about the solution of the relaxed minimization Problem 1:\\
\textbf{Statement~\ref{PofThm2}.1} : there exists $m\in \{1,\dots,N\}$, s.t.\\
1. for $1\le n\le m$, ${{d}_{n}}\ge \max \left( \frac{{{E}_{2}}}{{{E}_{1}}},1 \right)$, ${{t}_{n,1}}=0$;\\
2. for $m+1\le n\le N$, ${{d}_{n}}=0$, ${{t}_{n,2}}=0$.

Since ${{d}_{n}}\ge \max \left( \frac{{{E}_{2}}}{{{E}_{1}}},1 \right)$, we know that $m\max \left( \frac{{{E}_{2}}}{{{E}_{1}}},1 \right)\le D$. We can then rewrite the original optimization problem as follows:
\begin{equation}\label{min1}
  \begin{split}
    & \min_{\{t_{n,1},t_{n,2},d_n\}_{n=1}^N,\delta_1,\delta_2} E=\sum\limits_{n=1}^{m}{{{E}_{2}}{{t}_{n,2}}}\text{+}\sum\limits_{n=m+1}^{N}{{{E}_{1}}{{t}_{n,1}}}, \\
 & \text{s.t.}\left\{ \begin{matrix}
   \sum\limits_{n=1}^{m}{{{\epsilon }^{{{d}_{n}}{{t}_{n,2}}}}}\le {{\delta }_{2}},\sum\limits_{n=m+1}^{N}{{{\epsilon }^{{{t}_{n,1}}}}}\le {{\delta }_{1}},\sum\limits_{n=1}^{m}{{{d}_{n}}}\le D,  \\
   {{\delta }_{1}}+{{\delta }_{2}}<\ln \frac{1}{1-p_\text{tar}},{{\delta }_{1}},{{\delta }_{2}}\ge 0.  \\
   {{d}_{n}}\in [1,\infty ]\cup \{0\},{{t}_{n,1}},{{t}_{n,2}}\in [0,\infty ],\forall n
\end{matrix} \right. \\
  \end{split}
\end{equation}
When $m,{{\delta }_{1}}$ and ${{\delta }_{2}}$ are fixed, we decompose the problem into two sub-problems:
\begin{equation}
  \begin{split}
     \min \sum\limits_{n=1}^{m}{{{t}_{n,2}}},
  \text{  s.t. }\sum\limits_{n=1}^{m}{{{d}_{n}}}\le D,\sum\limits_{n=1}^{m}{{{\epsilon }^{{{d}_{n}}{{t}_{n,2}}}}}\le {{\delta }_{2}},{{\delta }_{2}}\ge 0.
  \end{split}
\end{equation}
\begin{equation}
  \begin{split}
  & \min \sum\limits_{n=m+1}^{N}{{{t}_{n,1}}}
 \text{  s.t. }\sum\limits_{n=m+1}^{N}{{{\epsilon }^{{{t}_{n,1}}}}}\le {{\delta }_{1}},{{\delta }_{1}}\ge 0.
\end{split}
\end{equation}
According to Lemma \ref{Theo2_lemma}, the first sub-problem, if $\frac{{{m}^{2}}}{\delta D}>{{e}^{1.5}}$, satisfies the lower bound
\begin{equation}
  \sum\limits_{n=1}^{m}{{{t}_{n,2}}}\ge \frac{{{m}^{2}}}{D\ln \left( 1/\epsilon  \right)}\ln \frac{m}{{{\delta }_{2}}}\ge \frac{{{m}^{2}}}{D\ln \left( 1/\epsilon  \right)}\ln \frac{m}{{{\delta }}},
\end{equation}
where
\begin{equation}\label{delt}
  \delta \text{=}\ln \frac{1}{1-p_\text{tar}}.
\end{equation}
The second sub-problem can be solved using simple convex-optimization techniques and the optimal solution satisfies
\begin{equation}
  \sum\limits_{n=m+1}^{N}{{{t}_{n,1}}}\ge \frac{N-m}{\ln \left( 1/\epsilon  \right)}\ln \frac{N-m}{{{\delta }_{1}}}\ge\frac{N-m}{\ln \left( 1/\epsilon  \right)}\ln \frac{N-m}{{{\delta }}}.
\end{equation}
Therefore, when $m$ is fixed,
\begin{equation}
  \begin{split}
    E=&\sum\limits_{n=1}^{N}{{{E}_{1}}{{t}_{n,1}}+{{E}_{2}}{{t}_{n,2}}}\\
    \ge &\frac{2{{m}^{2}}{{E}_{2}}}{D\ln \left( 1/\epsilon  \right)}\ln \frac{m}{\delta }\text{+}\frac{\left( N-m \right){{E}_{1}}}{\ln \left( 1/\epsilon  \right)}\ln \frac{N-m}{\delta },
  \end{split}
\end{equation}
If we choose $m\ge\frac{N}{2}$, and since $\frac{{{N}^{2}}}{4\delta D}>{{e}^{1.5}}$, we have that
\[E\ge \frac{{{N}^{2}}{{E}_{2}}}{4D\ln \left( 1/\epsilon  \right)}\ln \frac{N}{2\delta }=\Theta \left( \frac{{{N}^{2}}{{E}_{2}}}{D}\ln \frac{N}{\delta } \right).\]
If we choose $m<\frac{N}{2}$, we have that
\[E\ge \frac{N{{E}_{1}}}{2\ln \left( 1/\epsilon  \right)}\ln \frac{N}{2\delta }=\Theta \left( N{{E}_{1}}\ln \frac{N}{\delta } \right).\]
In the limit of small $p_\text{tar}$, $p_\text{tar}\approx \ln \frac{1}{1-p_\text{tar}}$. Thus, the minimization problem~\eqref{min1} always satisfies
\begin{equation}\label{p1lb_1}
  E=\Omega \left( \min \left( N{{E}_{1}}\ln \frac{N}{p_\text{tar}},\frac{{{N}^{2}}{{E}_{2}}}{D}\ln \frac{N}{p_\text{tar}} \right) \right).
\end{equation}
Moreover, the number of transmissions from distributed agents to the sink should be at least in the order of $N$, because there are $N$ bits to be transmitted over the binary erasure channels from the distributed agents to the sink. Therefore, $E\ge \Theta(NE_1)$, which, together with~\eqref{p1lb_1}, concludes that \eqref{p1lb} holds.

The lower bound of Problem 2 can be obtained similarly by relaxing Problem 2 to the following problem:
\begin{equation}
  \begin{split}
  & \min \sum\limits_{n=1}^{m}{{{d}_{n}}} \\
 & s.t.\left\{ \begin{matrix}
   \sum\limits_{n=1}^{m}{{{\epsilon }^{{{d}_{n}}{{t}_{n,2}}}}}\le {{\delta }_{2}},\sum\limits_{n=m+1}^{N}{{{\epsilon }^{{{t}_{n,1}}}}}\le {{\delta }_{1}},{{E}_{1}}{{T}_{1}}+{{E}_{2}}{{T}_{2}}\le {{E}_{M}},  \\
   \begin{matrix}
   {{\delta }_{1}}+{{\delta }_{2}}<\ln \frac{1}{1-{{p}_{\text{tar}}}},{{\delta }_{1}},{{\delta }_{2}}\ge 0,  \\
   \sum\limits_{n=1}^{m}{{{t}_{n,2}}}\le {{T}_{2}},\sum\limits_{n=m+1}^{N}{{{t}_{n,1}}}\le {{T}_{1}}  \\
\end{matrix}  \\
   {{d}_{n}}\in [0,\infty ],{{t}_{n,1}},{{t}_{n,2}}\in [1,\infty ]\cup \{0\},\forall n.  \\
\end{matrix} \right. \\
  \end{split}
\end{equation}
When $\delta_1,\delta_2$, $T_1,T_2$ and $m$ are fixed, the above problem can be decomposed into two sub-problems.
\begin{equation}\label{p66}
  \begin{split}
     \min \sum\limits_{n=1}^{m}{d_n},
  \text{  s.t. }\sum\limits_{n=1}^{m}{{t_{n,2}}}\le T_2,\sum\limits_{n=1}^{m}{{{\epsilon }^{{{d}_{n}}{{t}_{n,2}}}}}\le {{\delta }_{2}},{{\delta }_{2}}\ge 0.
  \end{split}
\end{equation}
\begin{equation}
  \begin{split}
  & \min 0,
 \text{  s.t. }\sum\limits_{n=m+1}^{N}{{t_{n,1}}}\le T_1,\sum\limits_{n=m+1}^{N}{{{\epsilon }^{{{t}_{n,1}}}}}\le {{\delta }_{1}},{{\delta }_{1}}\ge 0.
\end{split}
\end{equation}
Notice that the second sub-problem only tries to search for a feasible solution. Using convex programming techniques, we have that, when $\sum\limits_{n=m+1}^{N}{{{\epsilon }^{{{t}_{n,1}}}}}\le {{\delta }_{1}}$,
$\sum\limits_{n=m+1}^{N}{{t_{n,1}}}\ge \frac{N-m}{\ln (1/\epsilon)}\ln\frac{N-m}{\delta_1}$. Therefore, if $m<\frac{N}{2}$,
\[E_M>E_1T_1\ge\frac{N-m}{\ln (1/\epsilon)}\ln\frac{N-m}{\delta_1}>\frac{N}{2\ln (1/\epsilon)}\ln\frac{N}{2\delta},\]
which contradicts the condition $E_M< \frac{N{{E}_{1}}}{2\ln \left( 1/\epsilon  \right)}\ln \frac{N}{2\delta }$ in Problem 2. When $m\ge N/2$, it holds that $\frac{m^2}{\delta_2T_2}\ge\frac{N^2E_2}{4\delta E_M}>e^{1.5}$. Therefore, using Lemma~\ref{PofThm2}, we can also solve problem~\eqref{p66}. Skipping the details, we can show that the resulted optimization problem can be written as
\begin{equation}
\begin{split}
  & \min_{m\in\{N/2,\dots,N\},\delta_1,\delta_2,T_1,T_2} \frac{{{m}^{2}}}{{{T}_{2}}\ln \left( 1/\varepsilon  \right)}\ln \frac{m}{{{\delta }_{2}}}, \\
 & \text{s}\text{.t}\text{.}\left\{ \begin{matrix}
   \frac{N-m}{\ln \left( 1/\epsilon  \right)}\ln \frac{N-m}{{{\delta }_{1}}}\le {{T}_{1}},{{E}_{1}}{{T}_{1}}+{{E}_{2}}{{T}_{2}}\le {{E}_{M}},  \\
   {{\delta }_{1}}+{{\delta }_{2}}\le \ln \frac{1}{1-p_\text{tar}},{{\delta }_{1}},{{\delta }_{2}}\ge 0.
\end{matrix} \right.
\end{split}
\end{equation}
Noticing that ${{T}_{2}}\le {{E}_{M}}/{{E}_{2}}$, and hence that $\frac{{{m}^{2}}}{{{T}_{2}}\ln \left( 1/\varepsilon  \right)}\ln \frac{m}{{{\delta }_{2}}}\ge \frac{{{m}^{2}}E_2}{{{E}_{M}}\ln \left( 1/\varepsilon  \right)}\ln \frac{m}{\delta }$, where $\delta=\ln \frac{1}{1-p_\text{tar}} $, the solution of the above problem can be further lower-bounded by the solution of
\begin{equation}
\begin{split}
  & \min_{m\in\{N/2,\dots,N\},\delta_1,\delta_2,T_1,T_2} \frac{{{m}^{2}}E_2}{{{E}_{M}}\ln \left( 1/\varepsilon  \right)}\ln \frac{m}{\delta }, \\
 & \text{s}\text{.t}\text{.}\left\{ \begin{matrix}
   \frac{N-m}{\ln \left( 1/\epsilon  \right)}\ln \frac{N-m}{{{\delta }_{1}}}\le {{T}_{1}},{{E}_{1}}{{T}_{1}}+{{E}_{2}}{{T}_{2}}\le {{E}_{M}},  \\
   {{\delta }_{1}}+{{\delta }_{2}}\le \ln \frac{1}{1-p_\text{tar}},{{\delta }_{1}},{{\delta }_{2}}\ge 0,
\end{matrix} \right.
\end{split}
\end{equation}
which is equivalent to
\begin{equation}\label{p70}
\begin{split}
  & \min_{m\in\{N/2,\dots,N\}} \frac{{{m}^{2}}{{E}_{2}}}{\ln \left( 1/\varepsilon  \right){{E}_{M}}}\ln \frac{m}{\delta }, \\
 & \text{s.t.  }{{E}_{M}}\ge {{E}_{1}}\frac{N-m}{\ln \left( 1/\epsilon  \right)}\ln \frac{N-m}{\delta }.
\end{split}
\end{equation}
Since $m\ge N/2$, we know that the solution above satisfies $E^*\ge \frac{{{N}^{2}}{{E}_{2}}}{4\ln \left( 1/\varepsilon  \right){{E}_{M}}}\ln \frac{N}{2\delta }$.
\begin{remark}\label{remark2}
Notice that, if $E_M$ is very small, \emph{e.g.}, $E_M\to 0$, we have to set $m=N$ in~\eqref{p70}. The obtained lower bound has the form $|\mathcal{E}|\ge \frac{N^2E_2}{\ln E_M}\ln\frac{N}{\delta}\to \infty$ for a fixed $N$. This does not suggest that the lower bound is wrong, because in this case, the set of feasible solution in Problem 2 is empty. Therefore, the true minimization value of Problem 2 should be $\infty$, which means that the lower bound is still legitimate.
\end{remark}

\section{Proof of Lemma~\ref{identical_error}}\label{pf_Lemma2}
From Section~\ref{Algorithm}, we know that an error occurs when there exist more than one feasible solutions that satisfy the version with possible erasures of~\eqref{3g_code}. That is to say, when all positions with erasures are eliminated from the received vector, there are at least two solutions to the remaining linear equations. Denote by $\mathbf{x}_1$ and $\mathbf{x}_2$ two different vectors of self-information bits. We say that \emph{$\mathbf{x}_1$ is confused with $\mathbf{x}_2$} if the true vector of self-information bits is $\mathbf{x}_1$ but $\mathbf{x}_2$ also satisfies the possibly erased version of~\eqref{3g_code}, in which case $x_{1}$ is indistinguishable from $\mathbf{x}_2$. Denote by $P_e^{\mathcal{G}}(\mathbf{x}_1\rightarrow \mathbf{x}_2)$ the probability that $\mathbf{x}_1$ is confused with $\mathbf{x}_2$.

\begin{lemma}\label{all_zero_vector}
The probability that $\mathbf{x}_1$ is confused with $\mathbf{x}_2$ equals the probability that $\mathbf{x}_1-\mathbf{x}_2$ is confused with the $N$-dimensional zero vector $\mathbf{0}_N$, \emph{i.e.},
\begin{equation}\label{confusion_identical}
  P_e^{\mathcal{G}}(\mathbf{x}_1\rightarrow \mathbf{x}_2)=P_e^{\mathcal{G}}(\mathbf{x}_1-\mathbf{x}_2\rightarrow \mathbf{0}_N).
\end{equation}
\end{lemma}
\begin{proof}
We define an \emph{erasure matrix} $\mathbf{E}$ as a $2N$-by-$2N$ diagonal matrix in which each diagonal entry is either an `$e$' or a $1$. Define an extended binary multiplication operation with `$e$', which has the rule that $ae=e,a\in \{0,1\}$. The intuition is that both $0$ and $1$ become an erasure after being erased. Under this definition, the event that $\mathbf{x}_1$ is confused with $\mathbf{x}_2$ can be written as
\begin{equation}
  \mathbf{x}_1^\top\cdot [\mathbf{I},\mathbf{A}] \cdot \mathbf{E}=\mathbf{x}_2^\top\cdot [\mathbf{I},\mathbf{A}] \cdot \mathbf{E},
\end{equation}
\textcolor{black}{where a diagonal entry in $\mathbf{E}$ being `$e$' corresponds to erasure/removal of the corresponding linear equation.} We know that if the erasure matrix $\mathbf{E}$ remains the same, we can arrange the two terms and write
\begin{equation}
  (\mathbf{x}_1^\top-\mathbf{x}_2^\top)\cdot [\mathbf{I},\mathbf{A}] \cdot \mathbf{E}=0_N^\top\cdot [\mathbf{I},\mathbf{A}] \cdot \mathbf{E}.
\end{equation}
That is to say, if $\mathbf{x}_1$ is confused with $\mathbf{x}_2$, then, if all the erasure events are the same and the self-information bits are changed to $\mathbf{x}_1-\mathbf{x}_2$, they will be confused with the all zero vector $\mathbf{0}_N$ and vice-versa. Thus, in order to prove~\eqref{confusion_identical}, we only need to show that the probability of having particular erasure events remains the same with different self-information bits. This claim is satisfied, because by the BEC assumption the erasure events are independent of the channel inputs and identically distributed.
\end{proof}
Using the union bound, we have that
\begin{equation}
\begin{split}
  P_e^{\mathcal{G}}(\mathbf{x})\le \mathop \sum \limits_{\mathbf{x}_1^\top \in \{0,1\}^N \setminus \{\mathbf{x}\}} P_e^{\mathcal{G}}(\mathbf{x}\rightarrow \mathbf{x}_1).
\end{split}
\end{equation}
Thus, using the result from Lemma~\ref{all_zero_vector}, we obtain
\begin{equation}
\begin{split}
P_e^{\mathcal{G}}(\mathbf{x})\le \mathop \sum \limits_{\mathbf{x}_1^\top \in \{0,1\}^N \setminus \{\mathbf{x}\}} P_e^{\mathcal{G}}(\mathbf{x}-\mathbf{x}_1\rightarrow \mathbf{0}_N),
\end{split}
\end{equation}
which is equivalent to~\eqref{identical_error_equation}.
\section{Proof of Lemma~\ref{Lemma1}}\label{PofL1}
First, we notice that for $1\le i\le N$, the vector $\tilde{\mathbf{x}}^\top$ received is the noisy version of $\mathbf{x}_0^\top$. Since, according to the in-network computing scheme in Section~\ref{Algorithm}, the vector $\tilde{\mathbf{x}}^\top$ is obtained in the second step, the event $A_3^{(i)}(\mathbf{x}_0^\top)$ is the only ambiguity event. Moreover, if the $i$-th entry of $\mathbf{x}_0^\top$ is zero, it does not matter whether an erasure happens to this entry. Thus, the error probability can be calculated by considering all the $k$ non-zero entries, which means
\[
\mathop\prod \limits_{i=1}^{N}\Pr[A_1^{(i)}(\mathbf{x}_0^\top)\cup A_2^{(i)}(\mathbf{x}_0^\top)\cup A_3^{(i)}(\mathbf{x}_0^\top)]= \epsilon^k.
\]
For $N+1\le i\le 2N$, $A_3^{(i)}(\mathbf{x}_0^\top)$ is the erasure event during the second step and is independent from the previous two events $A_1^{(i)}(\mathbf{x}_0^\top)$ and $A_2^{(i)}(\mathbf{x}_0^\top)$. Therefore
\begin{equation}\label{Derive_1}
  \begin{split}
&\Pr\left[A_1^{(i)}(\mathbf{x}_{0}^{\top})\cup A_2^{(i)}(\mathbf{x}_{0}^{\top})\cup A_3^{(i)}(\mathbf{x}_{0}^{\top})\right]\\
\le& \Pr\left[(A_3^{(i)}(\mathbf{x}_{0}^{\top}))^C\right]+\Pr\left[A_3^{(i)}(\mathbf{x}_{0}^{\top})\right]\Pr\left[A_1^{(i)}(\mathbf{x}_{0}^{\top})\cup A_2^{(i)}(\mathbf{x}_{0}^{\top})\right]\\
=& 1-\epsilon+\epsilon\Pr\left[A_1^{(i)}(\mathbf{x}_{0}^{\top})\cup A_2^{(i)}(\mathbf{x}_{0}^{\top})\right]\\
=& 1-\epsilon+\epsilon\left(\Pr\left[A_1^{(i)}(\mathbf{x}_{0}^{\top})\right]+\Pr\left[(A_1^{(i)}(\mathbf{x}_{0}^{\top}))^C\cap A_2^{(i)}(\mathbf{x}_{0}^{\top})\right]\right).
\end{split}
\end{equation}
The event $A_1^{(i)}(\mathbf{x}_0^\top)$ happens when the local parity $\mathbf{x}_0^\top \mathbf{a}_{i}$ equals zero, \emph{i.e.}, in the $k$ locations of non-zero entries in $\mathbf{x}_0^\top$, there are an even number of ones in the corresponding entries in $\mathbf{a}_{i}$, the $i$-th column of the graph adjacency matrix $\mathbf{A}$. Denote by $l$ the number of ones in these $k$ corresponding entries in $\mathbf{a}_{i}$. Since each entry of $\mathbf{a}_{i}$ takes value 1 independently with probability $p$, the probability that an even number of entries are 1 in these $k$ locations is
\begin{equation}\label{Derive_2}
  \begin{split}
\Pr[A_1^{(i)}(\mathbf{x}_{0}^{\top})]=&\Pr[l\text{ is even}]\\
=&\mathop\sum\limits_{l\text{ is even}}p^l(1-p)^{k-l}=\frac{1+(1-2p)^k}{2}.
\end{split}
\end{equation}
The event $(A_1^{(i)}(\mathbf{x}_{0}^{\top}))^C\cap A_2^{(i)}(\mathbf{x}_{0}^{\top})$ indicates that $l$ is odd and at least one entry of all non-zero entries in $\mathbf{x}_0^\top$ is erased. Suppose in the remaining $N-k$ entries in $\mathbf{a}_{i}$, $j$ entries take the value 1 and hence there are $(l+j)$ 1's in $\mathbf{a}_{i}$. Therefore, for a fixed $l$, we have
\[\begin{split}
&\Pr[(A_1^{(i)}(\mathbf{x}_{0}^{\top}))^C\cap A_2^{(i)}(\mathbf{x}_{0}^{\top})|l]\\
=&\mathop\sum\limits_{j=0}^{N-k}\binom{N-k}{j}p^j (1-p)^{N-k-j}\cdot[1-(1-p_e)^{l+j}]\\
\le & \mathop\sum\limits_{j=0}^{N-k}\binom{N-k}{j}p^j (1-p)^{N-k-j}(l+j)p_e,
\end{split}\]
where $p$ is the edge connection probability and $p_e$ is the probability that a certain bit in $\mathbf{x}_0$ is erased for $t=\frac{\log(\frac{c\log N}{p_{\text{ch}}})}{\log(1/\epsilon)}$ times when transmitted to $v_i$ from one of its neighbors during the first step of the in-network computing scheme. Combining the above inequality with \eqref{one_bit_error_equation}, we get
\[\begin{split}
&\Pr[(A_1^{(i)})^C\cap A_2^{(i)}(l)]\\
\le& \mathop\sum\limits_{j=0}^{N-k}\binom{N-k}{j}p^j (1-p)^{N-k-j}(l+j)\frac{p_{\text{ch}}}{c\log N}\\
=&l\frac{p_{\text{ch}}}{c\log N}\mathop\sum\limits_{j=0}^{N-k}\binom{N-k}{j}p^j (1-p)^{N-k-j}\\
  &+\frac{p_{\text{ch}}}{c\log N}\mathop\sum\limits_{j=1}^{N-k}j\binom{N-k}{j}p^j (1-p)^{N-k-j}\\
\overset{(a)}{=}&l\frac{p_{\text{ch}}}{c\log N}\\
&+\frac{p_{\text{ch}}p}{c\log N}\mathop\sum\limits_{j=1}^{N-k}(N-k)\binom{N-k-1}{j-1}p^{j-1} (1-p)^{N-k-j}\\
=&l\frac{p_{\text{ch}}}{c\log N}+\frac{p_{\text{ch}}(N-k)}{N}\mathop\sum\limits_{j=1}^{N-k}\binom{N-k-1}{j-1}p^{j-1} (1-p)^{N-k-j}\\
=&l\frac{p_{\text{ch}}}{c\log N}+p_{\text{ch}}\cdot\frac{N-k}{N},
\end{split}\]
where step (a) follows from $j\binom{N-k}{j}=(N-k)\binom{N-k-1}{j-1}$. Therefore
\[\begin{split}
&\Pr[(A_1^{(i)})^C\cap A_2^{(i)}]\\
=&\mathop\sum\limits_{l\text{ is odd}}\binom{k}{l}p^l (1-p)^{k-l}\Pr[(A_1^{(i)})^C\cap A_2^{(i)}(l)]\\
\le&\mathop\sum\limits_{l\text{ is odd}}\binom{k}{l}p^l (1-p)^{k-l} (l\frac{p_{\text{ch}}}{c\log N}+p_{\text{ch}}\cdot\frac{N-k}{N})\\
=&\sum\limits_{l\text{ is odd}}{\binom{k}{l}{{p}^{l}}{{(1-p)}^{k-l}}{{p}_{\text{ch}}}\cdot \frac{N-k}{N}}\\
&+\sum\limits_{l\text{ is odd}}{l\binom{k}{l}{{p}^{l}}{{(1-p)}^{k-l}}\frac{{{p}_{\text{ch}}}}{c\log N}}\\
=&{{p}_{\text{ch}}}\cdot \frac{N-k}{N}\sum\limits_{l\text{ is odd}}{\binom{k}{l}}{{p}^{l}}{{(1-p)}^{k-l}}\\
&+\frac{kp{{p}_{\text{ch}}}}{c\log N}\sum\limits_{l\text{ is odd}}{\binom{k-1}{l-1}}{{p}^{l-1}}{{(1-p)}^{k-l}}\\
=&p_{\text{ch}}\cdot\frac{N-k}{N}\frac{1-(1-2p)^k}{2}+p_{\text{ch}}\cdot\frac{k}{N}\frac{1+(1-2p)^{k-1}}{2}\\
\overset{(a)}{\le}& Lp_{\text{ch}}\frac{1-(1-2p)^k}{2},
\end{split}\]
where the constant $L$ in step (a) is to be determined. Now we show that $L=\frac{2}{1-1/e}+1$ suffices to ensure that (a) holds. In fact, we only need to prove
\[\frac{N-k}{N}\frac{1-{{(1-2p)}^{k}}}{2}+\frac{k}{N}\frac{1+{{(1-2p)}^{k-1}}}{2}\le L\frac{1-{{(1-2p)}^{k}}}{2}.\]
Since $\frac{N-k}{N}<1$, it suffices to show that
\[\frac{k}{N}\frac{1+{{(1-2p)}^{k-1}}}{2}\le \left( L-1 \right)\frac{1-{{(1-2p)}^{k}}}{2}.\]
Since ${{(1-2p)}^{k-1}}<1$, it suffices to show that
\[\frac{k}{N}\le \left( L-1 \right)\frac{1-{{(1-2p)}^{k}}}{2},\]
or equivalently,
\begin{equation}\label{kkey}
  \frac{2k}{1-{{(1-2p)}^{k}}}\le N\left( L-1 \right).
\end{equation}
We know that
\[\begin{split}&1-{{(1-2p)}^{k}}\ge 2kp-C_{k}^{2}{{\left( 2p \right)}^{2}}\\
=&2kp-2k(k-1){{p}^{2}}=2kp\left[ 1-p(k-1) \right]\ge 2kp(1-kp).\end{split}\]
Thus, when $kp\le \frac{1}{2}$, $1-{{(1-2p)}^{k}}\ge 2kp(1-kp)\ge kp$ and
\[\frac{2k}{1-{{(1-2p)}^{k}}}\le \frac{2k}{kp}=\frac{2N}{c\log N}\le 2N,\] when $c\log N>1$.
When $kp>\frac{1}{2}$, ${{(1-2p)}^{k}}\le {{(1-2p)}^{\frac{1}{2p}}}\le \frac{1}{e}$ and
\[\frac{2k}{1-{{(1-2p)}^{k}}}\le \frac{2k}{1-1/e}\le \frac{2N}{1-1/e}.\]
Thus, as long as $L\ge 1+\frac{2}{1-1/e}$, \eqref{kkey} holds.
Jointly considering~\eqref{Derive_2}, we get
\[\Pr[A_1^{(i)}\cup A_2^{(i)}]\le\frac{1+(1-2p)^k}{2}+Lp_{\text{ch}}\frac{1-(1-2p)^k}{2}.\]
Combining~\eqref{Derive_1}, we finally arrive at
\[\begin{split}
&\Pr[A_1^{(i)}\cup A_2^{(i)}\cup A_3^{(i)}]\\
\le& \epsilon +(1-\epsilon )\left[ \frac{1+{{(1-2p)}^{k}}}{2}+L{{p}_{\text{ch}}}\frac{1-{{(1-2p)}^{k}}}{2} \right] \\
 & =\epsilon +(1-\epsilon )\left[ 1-\left( 1-L{{p}_{\text{ch}}} \right)\frac{1-{{(1-2p)}^{k}}}{2} \right] \\
 & =1-(1-\epsilon )\left( 1-L{{p}_{\text{ch}}} \right)\frac{1-{{(1-2p)}^{k}}}{2} \\
 & <1-(1-\epsilon -L{{p}_{\text{ch}}})\frac{1-{{(1-2p)}^{k}}}{2} \\
 & =1-(1-\epsilon -L{{p}_{\text{ch}}})\left[ 1-\frac{1+{{(1-2p)}^{k}}}{2} \right] \\
 & =\epsilon +L{{p}_{\text{ch}}}+(1-\epsilon -L{{p}_{\text{ch}}})\frac{1+{{(1-2p)}^{k}}}{2} \\
 & ={{\varepsilon }_{0}}+(1-{{\varepsilon }_{0}})\frac{1+{{(1-2p)}^{k}}}{2},
\end{split}\]
where $\varepsilon_0=Lp_{\text{ch}}+\epsilon$.
\section{Proof of Theorem~\ref{Theorem1}}\label{PofT1}
We will prove that for any $\eta>0$, it holds that
\begin{equation}\label{error_exponent_proof}
  P_e^{(N)} \le {(1 - b_{\eta})^N}{\rm{ + }}\eta e \epsilon \frac{{{N^{2 - c(1 - {\varepsilon _0})(1 - c\eta )}}}}{{\log N}}.
\end{equation}
As shown in what follows, we bound the right hand side of~\eqref{Error_Middle} with two different methods for different $k$'s. First, when $k$ satisfies
\begin{equation}\label{small_k}
  1\le k < \eta \frac{N}{{\log N}},
\end{equation}
define
\begin{equation}\label{u}
  u = N(1 - {\varepsilon _0})\frac{{1 - {{(1 - 2p)}^k}}}{2}
\end{equation}
Then, based on the inequality
\begin{equation}\label{exponential_ineq}
  {(1 - \frac{1}{x})^x} \le {e^{ - 1}},\forall x \in (0,1],
\end{equation}
we have
\begin{equation}\label{k_small_drv1}
\begin{split}
&{[{\varepsilon _0} + (1 - {\varepsilon _0})\frac{{1 + {{(1 - 2p)}^k}}}{2}]^N} \\
= &{(1 - \frac{u}{N})^N} = {[{(1 - \frac{u}{N})^{\frac{N}{u}}}]^u} \le {e^{ - u}}.
\end{split}
\end{equation}
From the Taylor's expansion, we get
\[{(1 - 2p)^k} = 1 - 2pk + \frac{{k(k - 1)}}{2}{\theta ^2},\theta  \in [0,2p].\]
By applying the equation above to~\eqref{u}, we get
\[u = N(1 - {\varepsilon _0})[kp - \frac{{k(k - 1)}}{4}{\theta ^2}].\]
Therefore, we have
\[\begin{split}
{e^{ - u}} = &{e^{ - k(1 - {\varepsilon _0}) \cdot c\log N}}\exp \{ N(1 - {\varepsilon _0})\frac{{k(k - 1)}}{4}{\theta ^2}\}\\
\le &{\left( {\frac{1}{N}} \right)^{ck(1 - {\varepsilon _0})}}\exp \{ N(1 - {\varepsilon _0})\frac{{k(k - 1)}}{4}\frac{{4{c^2}{{\log }^2}N}}{{{N^2}}}\}\\
 = &{\left( {\frac{1}{N}} \right)^{ck(1 - {\varepsilon _0})}}{N^{(1 - {\varepsilon _0}) \cdot \frac{{{c^2}k(k - 1)\log N}}{N}}}.
\end{split}\]
Plugging the above inequality into~\eqref{k_small_drv1}, we get
\begin{equation}
\begin{split}\label{k_small_drv2}
&\binom{N}{k}{\epsilon^k}{[{\varepsilon _0} + (1 - {\varepsilon _0})\frac{{1 + {{(1 - 2p)}^k}}}{2}]^N}\\
\le& {\left( {\frac{{Ne}}{k}} \right)^k}{\epsilon^k}{\left( {\frac{1}{N}} \right)^{ck(1 - {\varepsilon _0})}}{N^{(1 - {\varepsilon _0}) \cdot \frac{{{c^2}k(k - 1)\log N}}{N}}}\\
= &{\left( {\frac{e}{k}\epsilon{N^{1 - c(1 - {\varepsilon _0})[1 - \frac{{c(k - 1)\log N}}{N}]}}} \right)^k}<{\left( {\frac{e}{k}\epsilon{N^{1 - c(1 - {\varepsilon _0})(1 - c\eta )}}} \right)^k},
\end{split}
\end{equation}
where the last inequality follows from~\eqref{small_k}.

Second, when $k$ satisfies
\begin{equation}\label{big_k}
  k > \eta \frac{N}{{\log N}},
\end{equation}
we can directly write
\[{(1 - 2p)^k} = {[{(1 - 2p)^{\frac{1}{{2p}}}}]^{2pk}} \le {e^{ - 2pk}} <{e^{ - 2c\eta }}.\]
Therefore, it holds that
\[\begin{split}
&\sum\limits_{k > \eta \frac{N}{{\log N}}} {\binom{N}{k}{{\epsilon}^k}{{[{\varepsilon _0} + (1 - {\varepsilon _0})\frac{{1 + {{(1 - 2p)}^k}}}{2}]^N}}}\\
\le &\sum\limits_{k > \eta \frac{N}{{\log N}}} {\binom{N}{k}{{\epsilon}^k}{{[{\varepsilon _0} + (1 - {\varepsilon _0})\frac{{1 + {e^{ - 2c\eta }}}}{2}]^N}}}\\
\le &{[{\varepsilon _0} + (1 - {\varepsilon _0})\frac{{1 + {e^{ - 2c\eta }}}}{2}]^N}\sum\limits_{k = 0}^N {\binom{N}{k}{{\epsilon}^k}}\\
= &{[{\varepsilon _0} + (1 - {\varepsilon _0})\frac{{1 + {e^{ - 2c\eta }}}}{2}]^N}{(1 +\epsilon)^N}\\
= &{[(1 - (1 - {\varepsilon _0})\frac{{1 - {e^{ - 2c\eta }}}}{2})(1 +\epsilon)]^N}\\
\le &{\{ 1 - [(1 - {\varepsilon _0})(1 - \frac{{1 - {e^{ - 2c\eta }}}}{2}) - \epsilon]\} ^N}\\
=&{\{ 1 - (2b_\eta - \epsilon)\} ^N}.
\end{split}\]
When~\eqref{N_big_enough} holds, we have
\begin{equation}\label{sum_big_k}
\begin{split}
  &\sum\limits_{k > \eta \frac{N}{{\log N}}} {\binom{N}{k}{{(\frac{{{p_{\text{ch}}}}}{{c\log N}})}^k}{{[{\varepsilon _0} + (1 - {\varepsilon _0})\frac{{1 + {{(1 - 2p)}^k}}}{2}]^N}}}\\
  <& {(1 - b_\eta)^N}.
  \end{split}
\end{equation}
Combining~\eqref{Error_Middle} and~\eqref{k_small_drv2}, we get
\[\begin{split}
P_e^{(N)}&\le {(1 - b_\eta)^N}{\rm{+}}\\
&\sum\limits_{k < \eta \frac{N}{{\log N}}} {\binom{N}{k}{{\epsilon}^k}{{[{\varepsilon _0} + (1 - {\varepsilon _0})\frac{{1 + {{(1 - 2p)}^k}}}{2}]^N}}} \\
&\le{(1 - b_\eta)^N}+\sum\limits_{k < \eta \frac{N}{{\log N}}} {\left( {\frac{e}{k}\epsilon{N^{1 - c(1 - {\varepsilon _0})(1 - c\eta )}}} \right)^k}\\
&\le{(1 - b_\eta)^N}+ \eta \frac{N}{{\log N}}\frac{e}{k}\epsilon{N^{1 - c(1 - {\varepsilon _0})(1 - c\eta )}}\\
&\le {(1 - b_\eta)^N}{\rm{ + }}\eta e \epsilon \frac{{{N^{2 - c(1 - {\varepsilon _0})(1 - c\eta )}}}}{{\log N}}.
\end{split}\]
When $2 < c(1 - {\varepsilon _0})(1 - c\eta )$, the right hand side decreases polynomially with $N$.
\section{Proof of Corollary~\ref{coding_upb}}\label{pf_coding_upb}
The proof relies on building the relation between the $\mathcal{GC}$-3 graph code and an ordinary error control code. We construct the error control code as follows:
\begin{itemize}
  \item Construct a directed Erd$\ddot{o}$s-R$\acute{e}$nyi network $\mathcal{G}=(\mathcal{V},\mathcal{E})$ with $N$ nodes and connection probability $p=\frac{c\log N}{N}$, where $c$ is a constant which will be defined later.
  \item Construct a linear code with the generated matrix $\mathbf{G}=[\mathbf{I},\mathbf{A}]$, where $\mathbf{A}_{N\times N}$ is the adjacency matrix of the directed network in the previous step, \emph{i.e.}, the entry $A_{m,n}=1$ if and only if $v_m$ is connected to $v_n$.
\end{itemize}
The number of edges in $\mathcal{E}$ is a binomial random variable distributed according to $\text{Binomial}(N^2,p)$. Using the Chernoff bound~\cite{Che_AMS_52}, we obtain
\begin{equation}\label{Bin_large_deviation}
  \Pr(|\mathcal{E}|>2pN^2)<\exp(-\frac{p^2}{2}N^2)=(\frac{1}{N})^{\frac{c^2}{2}\log N}.
\end{equation}
Then we use the code constructed above to encode $N$ binary bits and transmit the encoded bits via $2N$ parallel BECs to the receiver. Denote by $A_e^{(N)}$ the event of a block error on the receiver side. \textcolor{black}{Define $P_e^{(N)}=\Pr(A_e^{(N)})$ as the block error probability. Note that
\begin{equation}
  P_e^{(N)}=\mathbb{E}\left[P_e^{\mathcal{G}}\right],
\end{equation}
where $P_e^{\mathcal{G}}=\Pr\left(A_e^{(N)}\mid\mathcal{G}\right)$ is the block error probability conditioned on the graph instance $\mathcal{G}$. In other words, $P_e^{(N)}$ is the expected block error probability of an ensemble of codes constructed based on directed Erd$\ddot{o}$s-R$\acute{e}$nyi networks.}

Clearly, this point-to-point transmitting scheme is the same as carrying out the in-network computing scheme in Section~\ref{Algorithm}, except that the encoding step in the point-to-point case is centralized instead of being distributed. This is equivalent to the in-network computing scheme when channels between neighboring sensor nodes are without erasures and erasures happen only when communicating over the channels to the decoder (compare with the second step of the in-network computing scheme). Since erasure events constitute a strict subset of those encountered in the in-network computing scheme, the upper bound on the error probability in Theorem~\ref{Theorem1} still holds, which means that the expected block error probability $P_e^{(N)}$ goes down polynomially when the constant $c$ designed for the connection probability $p=\frac{c\log N}{N}$ satisfies the same condition in Theorem~\ref{Theorem1}. Note that
\begin{equation}\label{rm_cd_argu}
\begin{split}
  P_e^{(N)}= \Pr(A_e^{(N)})=&\Pr(|\mathcal{E}|>2pN^2)\Pr\left(A_e^{(N)}\mid|\mathcal{E}|>2pN^2\right)\\
  &+\Pr(|\mathcal{E}|<2pN^2)\Pr\left(A_e^{(N)}\mid|\mathcal{E}|<2pN^2\right).
\end{split}
\end{equation}
Thus, combining~\eqref{rm_cd_argu} with~\eqref{Bin_large_deviation} and~\eqref{error_exponent}, we conclude that the block error probability conditioned on $|\mathcal{E}|<2pN^2$, or equivalently $\Pr(A_e^{(N)}||\mathcal{E}|<2pN^2)$, decreases polynomially with $N$. This means that, by expurgating the code ensemble and eliminating the codes that have more than $2pN^2=\mathcal{O}(N\log N)$ ones in their generator matrices, we obtain a sparse code ensemble, of which the expected error probability decreases polynomially with $N$. Therefore, there exists a series of sparse codes which obtains polynomially decaying error probability with $N$.
\end{document}